\documentclass[letterpaper,11pt]{paper}

\usepackage{fullpage}
\usepackage{amsthm}
\usepackage{amsmath}
\usepackage{amssymb}
\usepackage{latexsym}
\usepackage{epsfig}
\usepackage{cite}
\usepackage{psfrag}
\usepackage{pst-pdf}
\usepackage{ifdraft}
\usepackage{color}
\usepackage{paralist}
\usepackage{subfig}

\usepackage[colorlinks,urlcolor=blue,citecolor=blue,linkcolor=blue]{hyperref}

\def\doctype{2}

\def\tsubmission{1}

\ifnum\doctype=\tsubmission
	\newcommand{\full}[1]{}
	\newcommand{\submit}[1]{#1}
\else
	\newcommand{\full}[1]{#1}
	\newcommand{\submit}[1]{}
\fi

\newtheorem{theorem}{Theorem}[section]

\newtheorem{claim}[theorem]{Claim}
\newcommand{\claimproof}[2]%
{\noindent{\em Proof of Claim \ref{#1}.}
#2\hspace*{\fill}$\Box$~~~~~\vspace{5mm} }

\newtheorem{definition}[theorem]{Definition}

\newtheorem{lemma}[theorem]{Lemma}

\newtheorem{prop}[theorem]{Proposition}

\newcommand{\Sec}[1]{\hyperref[sec:#1]{\S\ref*{sec:#1}}} 
\newcommand{\Eqn}[1]{\hyperref[eq:#1]{(\ref*{eq:#1})}} 
\newcommand{\Fig}[1]{\hyperref[fig:#1]{Figure\,\ref*{fig:#1}}} 
\newcommand{\Tab}[1]{\hyperref[tab:#1]{Table\,\ref*{tab:#1}}} 
\newcommand{\Thm}[1]{\hyperref[thm:#1]{Theorem\,\ref*{thm:#1}}} 
\newcommand{\Lem}[1]{\hyperref[lem:#1]{Lemma\,\ref*{lem:#1}}} 
\newcommand{\Prop}[1]{\hyperref[prop:#1]{Prop.~\ref*{prop:#1}}} 
\newcommand{\Cor}[1]{\hyperref[cor:#1]{Corollary~\ref*{cor:#1}}} 
\newcommand{\Def}[1]{\hyperref[def:#1]{Definition~\ref*{def:#1}}} 
\newcommand{\Alg}[1]{\hyperref[alg:#1]{Alg.~\ref*{alg:#1}}} 
\newcommand{\Ex}[1]{\hyperref[ex:#1]{Ex.~\ref*{ex:#1}}} 
\newcommand{\Clm}[1]{\hyperref[clm:#1]{Claim~\ref*{clm:#1}}} 

\newcommand{\eqdef}{:=}
\newcommand{\etal}{et al.}

\newcommand{\vis}{\textrm{vis}}

\newcommand{\bS}{\textbf{S}}
\newcommand{\bT}{\textbf{T}}
\newcommand{\bU}{\textbf{U}}

\newcommand{\cD}{\mathcal{D}}
\newcommand{\cE}{\mathcal{E}}

\newcommand{\cF}{\mathcal{F}}

\newcommand{\cR}{\mathcal{R}}

\newcommand{\cT}{\mathcal{T}}

\newcommand{\eps}{\varepsilon}

\newcommand{\EX}{\hbox{\bf E}}

\newcommand{\R}{\mathbb{R}}

\newcommand{\OPT}{\text{OPT}}

\newcommand{\wlambda}{\widehat{\lambda}}

\newcommand{\wolfgang}[1][says]{WOLFGANG: {#1}}

\begin{document}
%
%
\title{Self-improving Algorithms for Coordinate-wise Maxima}

\author{
Kenneth L. Clarkson\thanks{%
  IBM Almaden Research Center, San Jose, USA.
  Email: {\tt klclarks@us.ibm.com}
}
\and
Wolfgang Mulzer\thanks{%
  Institut f\"ur Informatik, 
  Freie Universit\"at Berlin, Berlin, Germany.
  Email: {\tt mulzer@inf.fu-berlin.de}. 
}
\and
C. Seshadhri\thanks{%
Sandia National Labs, Livermore, USA.
Email: {\tt scomand@sandia.gov}
}
}

%
%
%
%
\maketitle
\begin{abstract}
Computing the coordinate-wise maxima of a planar point set
is a classic and well-studied problem in computational geometry. 
We give an algorithm for this problem in the \emph{self-improving setting}. 
We have $n$ (unknown) independent  distributions
$\cD_1, \cD_2, \ldots, \cD_n$ of planar points. An input pointset
$(p_1, p_2, \ldots, p_n)$
is generated by taking an independent sample $p_i$ from each $\cD_i$, so 
the input distribution $\cD$ is the product $\prod_i \cD_i$. A self-improving
algorithm repeatedly gets input sets from the distribution $\cD$ (which is 
\emph{a priori} unknown) and tries to
optimize its running time for $\cD$. Our algorithm uses the first few
inputs to learn salient features of the distribution, and then becomes an 
optimal algorithm 
for distribution $\cD$. Let $\OPT_\cD$ denote the expected depth of 
an \emph{optimal} linear comparison tree computing the maxima for 
distribution $\cD$. Our algorithm
eventually has an expected running time of $O(\text{OPT}_\cD + n)$, 
even though it did not know $\cD$ to begin with.

Our result requires new tools to understand linear comparison trees 
for computing maxima. 
We show how to convert general linear comparison trees to very restricted 
versions, which can then be
related to the running time of our algorithm. An interesting feature of our 
algorithm
is an interleaved search, where the algorithm tries to determine the likeliest
point to be maximal with minimal computation. This allows the running time to be
truly optimal for the distribution $\cD$.
\end{abstract}




\section{Introduction}\label{sec:intro}

Given a set $P$ of $n$ points in the plane, the \emph{maxima}
problem is to find those points $p\in P$
for which no other point in $P$ has a larger $x$-coordinate and a larger
$y$-coordinate. More formally, for $p\in\R^2$, let $x(p)$ and $y(p)$ denote the
$x$ and $y$ coordinates of $p$. Then $p'$ \emph{dominates} $p$ if and only 
if $x(p')\ge
x(p)$, $y(p')\ge y(p)$, and one of these inequalities is strict. The desired
points are those in $P$ that are not dominated by any other points in $P$.
The set of maxima is also known as a \emph{skyline} 
in the database literature \cite{skyline} and as a \emph{Pareto frontier}.

This algorithmic problem has been studied since at least 1975 \cite{KLP}, when 
Kung~\etal~described an algorithm with an $O(n\log n)$ worst-case time and 
gave an $\Omega(n\log n)$
lower bound. Results since
then include average-case running times of $n+\tilde{O}(n^{6/7})$ point-wise
comparisons \cite{Golin}; output-sensitive algorithms needing $O(n\log h)$ time
when there are $h$ maxima \cite{KS}; and algorithms operating in
external-memory models \cite{GTVV}. 
A major problem with worst-case analysis is that it may not reflect the behavior
of real-world inputs. Worst-case algorithms are tailor-made for extreme
inputs, none of which may occur (with reasonable frequency) in practice. Average-case
analysis tries to address this problem by assuming some fixed distribution on inputs;
for maxima, the property of coordinate-wise independence covers a broad
range of inputs, and allows a clean analysis~\cite{Buchta}, but is 
unrealistic even so.
The right distribution to analyze remains a point of investigation.
Nonetheless, the assumption of randomly
distributed inputs is very natural and one worthy of further research.\\

{\bf The self-improving model.} Ailon~\etal~introduced the self-improving model
 to address this issue \cite{ACCL}. In this model, there is some fixed but unknown input distribution
$\cD$ that generates independent inputs, that is, whole input sets $P$. The algorithm initially
undergoes a \emph{learning phase}, where it processes inputs with a worst-case
guarantee but tries to learn information about $\cD$. The aim
of the algorithm is to become optimal \emph{for the distribution $\cD$}.
After seeing some (hopefully small) number of inputs, the algorithm shifts
into the \emph{limiting phase}. Now, the algorithm is tuned for $\cD$
and the expected running time is (ideally) optimal for $\cD$. A self-improving
algorithm can be thought of as an algorithm that attains the optimal average-case
running time for all, or at least a large class of, distributions $\cD$. 

Following earlier self-improving algorithms, we assume the input has
a product distribution. An input 
is a set of $n$ points $P = (p_1, p_2, \ldots, p_n)$ in the plane. Each $p_i$ is generated
independently from a distribution $\cD_i$, so the probability distribution
of $P$ is the product $\prod_i \cD_i$. The $\cD_i$s themselves
are arbitrary, and the only assumption made is their independence.
There are lower bounds showing that some restriction on $\cD$ is
necessary for a reasonable self-improving algorithm, as we explain later.

The first self-improving algorithm was for sorting;
this was extended to Delaunay triangulations,
with these results eventually merged \cite{CS_self_improve,AilonCCLMS11}. A
self-improving algorithm for planar convex hulls was given by 
Clarkson~\etal~\cite{CMS_self_improve}, however their analysis was recently
discovered to be flawed.

\subsection{Main result} \label{sec:result}

Our main result is a self-improving algorithm for planar coordinate-wise maxima
over product distributions. We need some basic definitions before stating 
our main theorem. We explain what it means for a maxima algorithm
to be optimal for a distribution $\cD$. This in turn requires
a notion of \emph{certificates} for maxima, which allow the correctness
of the output to be verified in $O(n)$ time.
Any procedure for computing maxima must provide some ``reason"
to deem an input point $p$ non-maximal. The simplest certificate would be to
provide an input point dominating $p$. Most current algorithms
implicitly give exactly such certificates \cite{KLP,Golin,KS}. 


\begin{definition}\label{def:cert}
A \emph{certificate} $\gamma$ has:
\textup(i\textup) the sequence of the indices of the maximal points, sorted from 
left to right; 
\textup(ii\textup) for each non-maximal point, a \emph{per-point certificate}
of non-maximality, which is simply the index of an input point that 
dominates it. 
We say that  a certificate $\gamma$ is \emph{valid} for an input $P$ 
if $\gamma$
satisfies these conditions for $P$.
\end{definition}

The model of computation that we use to define optimality is
a linear computation tree that generates query lines using
the input points. In particular,
our model includes the usual CCW-test that forms the
basis for many geometric algorithms. 

Let $\ell$ be a directed line. We use $\ell^+$ to denote the 
open halfplane
to the left of $\ell$ and $\ell^-$ to denote the open halfplane to
the right of $\ell$. 

\begin{definition} \label{def:opt} 
A \emph{linear comparison tree} $\cT$ is a binary
tree such that each node $v$ of $\cT$ is labeled with a query of the
form ``$p \in \ell_v^+?$''. Here $p$ denotes an input point and
$\ell_v$ denotes a directed line. The line $\ell_v$ can be obtained
in three ways:
\textup(i\textup) it can be a line independent of the input \textup(but dependent
on the node $v$\textup);
\textup(ii\textup) it can be a line with a slope independent of the 
input \textup(but dependent on $v$\textup) passing through a given
input point;
\textup(iii\textup) it can be a line through an input point and through a 
point $q$ independent of the
input \textup(but dependent on $v$\textup);
\textup(iv\textup) it can be the line defined by two
distinct input points. 
A linear comparison tree is \emph{restricted} if it only makes
queries of type \textup(i\textup).

A linear comparison tree $\cT$ \emph{computes the maxima}
for $P$ if each leaf corresponds to a certificate. This means that
each leaf $v$ of $\cT$ is labeled with a certificate $\gamma$ that
is valid for every possible input $P$ that reaches $v$.
\end{definition}

Let $\cT$ be a linear comparison tree and $v$ be a node of $\cT$. 
Note that $v$ corresponds to a 
region $\mathcal{R}_v \subseteq \R^{2n}$ such that an evaluation of $\cT$ 
on input $P$ reaches $v$ if and only if $P \in \mathcal{R}_v$. 
If $\cT$ is \emph{restricted}, then $\mathcal{R}_v$
is the Cartesian product of a sequence
$ (R_1, R_2, \ldots, R_n)$ of polygonal regions.
The \emph{depth}
of $v$, denoted by $d_v$, is the length of the path from the root of $\cT$ 
to $v$. Given $\cT$, there exists exactly one leaf $v(P)$ that is reached 
by the evaluation of $\cT$ on input $P$. The \emph{expected depth} 
of $\cT$ over $\cD$, $d_\cD(\cT)$, 
is defined as $\EX_{P \sim \cD}[d_{v(P)}]$. 
Consider some comparison based algorithm $A$ that is modeled
by such a tree $\cT$. The expected depth of $\cT$ is a lower
bound on the number of comparisons performed by $A$.

Let $\bT$ be the set of trees that compute the maxima of $n$ points.
We define $\OPT_\cD = \inf_{\cT \in \bT} d_\cD(\cT)$. This is
a lower bound on the expected time taken by \emph{any} linear
comparison tree to compute the maxima of inputs distributed according
to $\cD$. We would like our algorithm to have a running time comparable 
to $\OPT_\cD$.

%
%
\begin{theorem} \label{thm:main} 
Let $\eps > 0$ be a fixed constant and $\cD_1$, $\cD_2$, $\ldots, \cD_n$ be
independent planar point distributions. The input distribution is $\cD = \prod_i \cD_i$.
There is a self-improving algorithm to compute the coordinate-wise maxima 
whose expected time in the limiting phase is $O(\eps^{-1}(n + \OPT_\cD))$. 
The learning phase lasts for $O(n^\eps)$ inputs and the space requirement is $O(n^{1+\eps})$.
\end{theorem}

There are lower bounds in \cite{AilonCCLMS11} (for sorters) implying that a self-improving
maxima algorithm that works for all distributions requires exponential storage, and
that the time-space tradeoff (wrt $\eps$) in the above theorem is optimal.

\begin{figure*}[tb!]
\begin{center}
  \includegraphics[scale=01]{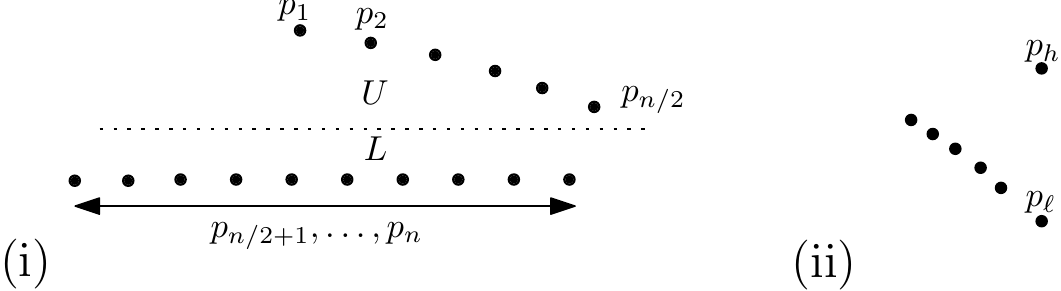} 
\end{center}
\caption{Examples of difficult distributions}
\label{fig:bad}
\end{figure*}

\textbf{Challenges.} One might think that since self-improving sorters
are known, an algorithm for maxima should follow directly. But this reduction
is only valid for $O(n\log n)$ algorithms. Consider Figure~\ref{fig:bad}(i). 
The distributions $\cD_1$, $\cD_2$, $\ldots, \cD_{n/2}$
generate the fixed points shown. The remaining distributions generate a random
point from a line below $L$. Observe that an algorithm that wishes to sort
the $x$-coordinates requires $\Omega(n\log n)$ time. On the other hand, there is a simple
comparison tree that determines the maxima in $O(n)$ time. For all $p_j$ where  $j > n/2$,
the tree simply checks if $p_{n/2}$ dominates $p_j$. After that, it performs a linear
scan and outputs a certificate.

We stress that even though the points are independent, the collection 
of maxima exhibits strong dependencies.
In Figure~\ref{fig:bad}(ii), suppose a distribution $\cD_i$ generates
either $p_h$ or $p_\ell$; if $p_\ell$ is chosen, we must consider the 
dominance relations among the remaining points,
while if $p_h$ is chosen, no such evaluation is required.
The optimal search tree for a distribution $\cD$ must exploit this  complex dependency.

Indeed, arguing about optimality is one of the key contributions of this work. Previous
self-improving algorithms employed information-theoretic optimality arguments. These
are extremely difficult to analyze for settings like maxima, where some points are more 
important to process that others, as in Figure~\ref{fig:bad}. (The main error
in the self-improving convex hull paper  \cite{CMS_self_improve}
was an incorrect consideration of dependencies.)
We focus on a somewhat weaker notion of optimality---linear comparison 
trees---that
nonetheless covers most (if not all) important algorithms for maxima.

In Section~\ref{sec:model}, we describe how to convert linear comparison trees into
restricted forms that use much more structured (and simpler) queries.
Restricted trees are much more amenable to analysis. In some
sense, a restricted tree decouples the individual
input points and makes the maxima computation amenable to separate
$\cD_i$-optimal searches.
A leaf of a restricted tree is associated with a sequence of polygons
$(R_1,R_2,\ldots, R_n)$ such that the leaf is visited if and only if every $p_i\in R_i$,
and conditioned on that event, the $p_i$ remain independent.
This independence is extremely important for the analysis.
We design an algorithm whose
behavior can be related to the restricted tree. Intuitively, if the algorithm spends
many comparisons involving a single point, then we can argue that the optimal
restricted tree must also do the same. We give more details about the algorithm
in Section~\ref{sec:outline}.

\subsection{Previous work} \label{sec:prev}
Afshani~\etal~\cite{AfshaniBC09} introduced a model of \emph{instance-optimality}
applying to algorithmic problems including planar convex hulls and
maxima. (However, their model is different from, and in a sense weaker than,
the prior notion of instance-optimality introduced by Fagin~\etal~\cite{FLM}.)
All previous (e.g., output sensitive and instance optimal) algorithms require expected $\Omega(n\log n)$ time for the distribution 
given in Figure~\ref{fig:bad}, though an optimal self-improving algorithm
only requires $O(n)$ expected time. (This was also discussed in \cite{CMS_self_improve} with a similar example.), 

We also mention the paradigm of preprocessing regions 
in order to compute certain geometric
structures 
faster (see, e.g., 
\cite{BuchinLoMoMu11,EzraMu11,HeldMi08,LoefflerSn10,vKreveldLoMi10}). 
Here, we are given a set $\mathcal{R}$ of planar regions, and we  
would like to preprocess $\mathcal{R}$ 
in order to quickly find the (Delaunay) triangulation (or convex hull) 
for any point set 
which contains exactly one point from each region in $\mathcal{R}$.  
This setting is 
adversarial, but if we only consider point sets where a point is randomly 
drawn from each region, it can be regarded as a special case of our 
setting. In this view, these results give us bounds on the running
time a self-improving algorithm can achieve if $\mathcal{D}$ draws 
its points from disjoint planar regions. 

\subsection{Preliminaries and notation}

Before we begin, let us define some basic concepts and agree on a
few notational conventions. 
We use $c$ for a sufficiently large constant, and we 
write $\log x$ to denote the logarithm of $x$ in base 2.
All the probability distributions are assumed to be
continuous. (It is not necessary to do this, but it makes
many calculations a lot simpler.)

Given a polygonal region $R \subseteq \R^2$ and a 
probability distribution $\cD$ on the plane, we call $\ell$ a
\emph{halving line} for $R$ (with respect to $\cD$) if
\[
\Pr_{p \sim \cD}[p \in \ell^+ \cap R]
= 
\Pr_{p \sim \cD}[p \in \ell^- \cap R].
\]
Note that if $\Pr_{p \sim \cD}[p \in R] = 0$, 
every line is a halving line for $R$. If not, a halving line
exactly halves the conditional probability for $p$ being in
each of the corresponding halfplanes, conditioned on $p$ lying
inside $R$.

Define a \emph{vertical slab structure} $\bS$ as
a sequence of vertical lines partitioning the 
plane into vertical regions, called \emph{leaf slabs}. 
(We will consider the latter to be the open regions between the vertical lines.
Since we assume that our distributions are continuous, we abuse
notation and consider the leaf slabs to partition the plane.)
More generally, a \emph{slab} is the region between
any two vertical lines of the $\bS$.
The \emph{size} of the slab structure is the number of
leaf slabs it contains. We denote it by $|\bS|$. 
Furthermore, for any slab $S$, the probability that $p_i \sim \cD_i$
is in $S$ is denoted by $q(i,S)$.

A \emph{search tree} $T$ over $\bS$ is a comparison tree that locates a point 
within leaf slabs of $\bS$. Each internal node compares the $x$-coordinate of 
the point with a vertical line of $\bS$, and moves left or right accordingly.
We associate each internal node $v$ with a slab $S_v$ (any point in $S_v$
will encounter $v$ along its search).

\subsection{Tools from self-improving algorithms} \label{sec:tools}

We introduce some tools that were developed in previous self-improving results.
The ideas are by and large old, but our presentation in this form
is new. We feel that the following statements (especially \Lem{search-time})
are of independent interest. 

We define the notion of \emph{restricted searches},
introduced in \cite{CMS_self_improve}. This notion is central to our final
optimality proof. (The lemma and formulation as given here are new.)
Let $\bU$ be an ordered set and $\cF$ be a distribution
over $\bU$. For any element $j \in \bU$, $q_j$
is the probability of $j$ according to $\cF$. For any interval
$S$ of $\bU$, the total probability of $S$ is $q_S$.

We let $T$ denote a search tree over $\bU$. It will be convenient
to think of $T$ as (at most) ternary, where each node has at most $2$ children
that are internal nodes.
In our
application of the lemma, $\bU$ will just be the set
of leaf slabs of a slab structure $\bS$. 
We now introduce some definitions regarding restricted searches and search trees.


\begin{definition} \label{def:rest}
Consider a distribution $\cF$ and an interval $S$ of $U$. 
An \emph{$S$-restricted distribution} is given by
the probabilities (for element $r \in U$) $q'_r/\sum_{j \in U} q'_j$, where
the sequence $\{q'_j | j \in U\}$ has the following property.
For each $j \in S$, $0 \leq q'_j \leq q_j$.
For every other $j$, $q'_j = 0$.  

Suppose $j \in S$. An \emph{$S$-restricted search} is a search for $j$ in $T$
that terminates once $j$ is located in any interval contained in $S$.

\end{definition}

For any sequence of numbers $\{q'_j | j \in U\}$ and $S \subseteq U$, we use $q'_S$
to denote $\sum_{j \in S} q'_j$.

\begin{definition} \label{def:tree}
Let $\mu \in (0,1)$ be a parameter.
A search tree $T$ over $\bU$ is $\mu$-reducing if: for any internal node 
$S$ and for any non-leaf child $S'$ of $S$,
$q_{S'} \leq \mu q_S$. 

A search tree $T$ is \emph{$c$-optimal for restricted searches over $\cF$} if: for all
$S$ and $S$-restricted distributions $\cF_S$, the expected time of an $S$-restricted 
search over $\cF_S$ is at most $c(-\log q'_S + 1)$. (The probabilities $q'$ are
as given in \Def{rest}.)
\end{definition}

We give the main lemma about restricted searches. 
A tree that is optimal for searches over $\cF$ also works for
restricted distributions. 
\submit{The proof is given in the full version of the paper.}
\full{The proof is given in Appendix~\ref{app:restricted}.}

\begin{lemma} \label{lem:search-time} Suppose $T$
is a $\mu$-reducing search tree for $\cF$.
Then $T$ is $O(1/\log(1/\mu))$-optimal for restricted searches over $\cF$.
\end{lemma}

We list theorems about data structures that are built
in the learning phase. 
Similar structures were first constructed in \cite{AilonCCLMS11}, and the
following can be proved using their ideas. 
The data structures involve construction of slab structures and 
specialized search trees for each distribution $\cD_i$.
It is also important that these trees can be represented
in small space, to satisfy the requirements of \Thm{main}.
The following lemmas give us the details of the data structures required.
Because this is not a major contribution of this paper,
we relegate the details to \S\ref{sec:learning}.

\begin{lemma} \label{lem:slabstruct}
We can construct a slab structure $\bS$ with $O(n)$ leaf slabs 
such that, with probability $1-n^{-3}$ over the construction of 
$\bS$, the following holds.  For a leaf slab $\lambda$ of $\bS$, 
let $X_\lambda$ denote the number of points in a random 
input $P$ that fall into $\lambda$.  For every leaf slab 
$\lambda$ of $\bS$, we have $\EX[X^2_\lambda] = O(1)$.  
The construction takes $O(\log n)$ rounds and $O(n\log^2 n)$ time.
\end{lemma}

\begin{lemma} \label{lem:tree} Let $\eps > 0$ be a fixed parameter.
In $O(n^{\eps})$ rounds and $O(n^{1+\eps})$ time, we can construct search trees
$T_1$, $T_2$, $\ldots$, $T_n$ over $\bS$ such that the following holds.
\textup(i\textup) the trees can be represented in $O(n^{1+\eps})$ total space;
\textup(ii\textup) with probability $1-n^{-3}$ over the construction of the
$T_i$s, every $T_i$ is $O(1/\eps)$-optimal for restricted searches over $\cD_i$.
%
\end{lemma}

%
\section{Outline} \label{sec:outline}

We start by providing a very informal overview of the algorithm.
Then, we shall explain how the optimality is shown.

If the points of $P$ are sorted by $x$-coordinate, the
maxima of $P$ can be found easily by a
right-to-left sweep over $P$:
we maintain the largest $y$-coordinate $Y$ of the points
traversed so far; when a point $p$ is visited in the traversal,
if $y(p)<Y$, then $p$ is non-maximal, and the point $p_j$ with
$Y=y(p_j)$ gives a per-point certificate for $p$'s non-maximality.
If $y(p)\ge Y$, then $p$ is maximal, and can be put at the beginning
of the certificate list of maxima of $P$.

This suggests the following approach to a self-improving algorithm for maxima:
sort $P$ with a self-improving sorter and then use the traversal.
The self-improving sorter of \cite{AilonCCLMS11}
works by locating each point of $P$ within the slab structure
$\bS$ of \Lem{slabstruct} using the trees $T_i$ of
\Lem{tree}.


While this approach does use $\bS$ and the $T_i$'s,
it is not optimal for maxima,
because the time spent finding the exact sorted order
of non-maximal points may be wasted: %
in some sense, we are learning much more information about the input $P$ than necessary. To
deduce the list of maxima, we do not need 
the sorted order of
\emph{all} points of $P$: it suffices to know the sorted order of just
the maxima!
An optimal algorithm would probably
locate the maximal points in $\bS$ and would not bother locating
``extremely non-maximal'' points. This is, in some sense, the difficulty 
that output-sensitive algorithms face. 

As a thought experiment, let us suppose that the maximal
points of $P$ are known to us, but not in sorted order. We search only for these 
in $\bS$ and determine the sorted
list of maximal points. 
We can argue that the optimal algorithm must also
(in essence) perform such a search. We also need to find per-point
certificates for the non-maximal points.
We use the slab structure $\bS$ and the search trees, but now we shall
be very conservative in our searches.
Consider the search for a point $p_i$. At any intermediate
stage of the search, $p_i$ is placed in a slab $S$.
This rough knowledge of $p_i$'s location 
may already suffice to certify 
its non-maximality: 
let $m$ denote the leftmost maximal
point to the right of $S$ (since the sorted list of maxima is known, this information
can be easily deduced). We check if $m$ dominates $p_i$. If so, we have a per-point certificate
for $p_i$ and we promptly terminate the search for $p_i$. Otherwise,
we continue the search by a single step and repeat. 
We expect that many searches will not proceed too long, achieving
a better position to compete with the optimal algorithm.

Non-maximal points that are dominated by many maximal points will usually
have a very short search. Points
that are ``nearly'' maximal will require a much longer search. 
So this approach
should derive just the ``right" amount of information to 
determine the maxima output.
But wait!  Didn't we assume that the maximal points were known? 
Wasn't this crucial
in cutting down the search time? This is too much of an assumption, and 
because the maxima are highly dependent on each other, it is not clear
how to determine which points are maximal before performing searches.

The final algorithm overcomes this difficulty by interleaving
the searches for sorting the points with confirmation
of the maximality of some points, in a rough right-to-left order
that is a more elaborate version of the traversal scheme given
above for sorted points. 
The searches for all points $p_i$ (in their respective
trees $T_i$) are performed ``together", and their order is carefully chosen. At any intermediate
stage, each point $p_i$ is located in some slab $S_i$, represented by
some node of its search tree. We choose a specific point and advance
its search by one step. This order is very important, and is the basis
of our optimality. The algorithm is described in detail and analyzed in \S\ref{sec:algorithm}.

\textbf{Arguing about optimality.} A major challenge of self-improving
algorithms is the strong requirement of optimality for the distribution $\cD$.
We focus on the model of linear comparison trees, and let $\cT$ be an optimal
tree for distribution $\cD$. (There may be distributions where such an exact $\cT$
does not exist, but we can always find one that is near optimal.) One of 
our key insights is that when $\cD$ is a product distribution, then we
can convert $\cT$ to $\cT'$, a restricted comparison tree whose expected
depth is only a constant factor worse. In other words, there exists a near optimal
restricted comparison tree that computes the maxima. 

In such a tree, a leaf is labeled with a sequence of regions
$\cR = (R_1, R_2, \ldots, R_n)$.
Any input $P = (p_1, p_2, \ldots, p_n)$ such that $p_i \in R_i$ for all $i$,
will lead to this leaf. Since the distributions are independent, we
can argue that the probability that an input leads to this leaf
is $\prod_i \Pr_{p_i \sim \cD_i}[p_i \in R_i]$. Furthermore,
the depth of this leaf can be shown to be $-\sum_i \log \Pr[p_i \in R_i]$.
This gives us a concrete bound that we can exploit.

It now remains to show that if we start with a random input from $\cR$,
the expected running time is bounded by the sum given above. We will
argue that for such an input, as soon as the search for $p_i$ locates
it inside $R_i$, the search will terminate. This leads to the optimal
running time.

\section{The Computational Model and Lower Bounds}\label{sec:model}

\subsection{Reducing to restricted comparison trees} \label{sec:restrict}

We prove that when $P$ is generated
probabilistically, it suffices to focus on restricted comparison 
trees. To show this, we provide a sequence of transformations, starting
from the more general comparison tree,
that results in a restricted linear comparison tree of comparable
expected depth.
The main lemma of this section
is the following. 

\begin{lemma}\label{lem:restrictedTree}
Let $\cT$ a finite linear comparison tree and $\cD$ be a product
distribution over points. Then there exists a restricted 
comparison tree $\cT'$ with expected depth $d_{\cD}(\cT') = O(d_\cD(\cT))$, as 
$d_\cD(\cT) \rightarrow\infty$.
\end{lemma}

We will 
describe a transformation from  $\cT$ into a restricted comparison tree
with similar depth. The first step is to show how to represent a single 
comparison by a restricted linear comparison tree, provided that $P$ is
drawn from a product distribution. The final transformation basically
replaces each node of $\cT$ by the subtree given by the next claim.
For convenience, we will drop the subscript of $\cD$ from $d_\cD$, since
we only focus on a fixed distribution.

\begin{claim} \label{clm:node} 
Consider a comparison $C$ as described
in \Def{opt}, where the comparisons are listed in increasing order of simplicity.
Let
$\cD'$ be a product distribution for $P$ such that each $p_i$
is drawn from a polygonal region $R_i$. Then either $C$ is the
simplest, type \textup(i\textup) comparison, or there exists
a restricted linear comparison tree $\cT'_C$ that resolves the comparison $C$
such that the expected depth of $\cT'_C$ (over the distribution $\cD'$)
is $O(1)$, and all comparisons used in $\cT'_C$ are simpler than $C$.
\end{claim}

\begin{proof}
\noindent\textbf{$v$ is of type (ii).} 
  This means that $v$ needs to determine whether an input point $p_i$ lies 
  to the left of the directed line $\ell$ through another input point $p_j$ 
  with a fixed slope $a$. We replace
   this comparison with a binary search. Let $R_j$ be the region in 
   $\cD'$ corresponding to $p_j$.  Take a halving line $\ell_1$ 
   for $R_j$
   with slope~$a$. Then perform two comparisons to determine on which side of
   $\ell_1$ the inputs $p_i$ and $p_j$ lie. If $p_i$ and $p_j$ lie on different
   sides of $\ell_1$, we declare success and resolve the original comparison
   accordingly. Otherwise, we replace $R_j$ with the appropriate new region and
   repeat the process until we can declare success.
   Note that in each attempt the success probability is at least $1/4$.
   The resulting restricted tree $\cT'_C$ can be infinite. 
   Nonetheless, the probability that an evaluation of $\cT'_C$
   leads to a node of depth $k$ is at most $2^{-\Omega(k)}$, so the
   expected depth is $O(1)$. 
   
\vspace{5pt}

\noindent\textbf{$v$ is of type (iii).} 
  Here the node $v$ needs to determine whether an input point $p_i$ lies 
  to the left of the directed line $\ell$ through another input point $p_j$ 
  and a fixed point $q$.
  
  We partition the plane by a constant-sized family of cones,
  each with apex $q$, such that for each cone $V$ in the family,
  the probability that line $\overline{q p_j}$ meets $V$ (other than at $q$)
  is at most $1/2$. Such a family
  could be constructed by a sweeping a line around $q$, or by
  taking a sufficiently large, but constant-sized, sample from the distribution
  of $p_j$, and bounding the cones by all lines through $q$ and each
  point of the sample. Such a construction has a non-zero probability
  of success, and therefore the described family of cones exists.
  
  We build a restricted tree that locates a point in the corresponding cone.
  For each cone $V$, we can recursively build such a family of cones (inside $V$),
  and build a tree for this structure as well. Repeating for each cone, this
  leads to an infinite restricted tree $\cT'_C$. We search for both $p_i$ and $p_j$
  in $\cT'_C$. When we locate $p_i$ and $p_j$ in two different cones of the same family,
  then comparison between $p_i$ and $\overline{q p_j}$ is resolved and the search
  terminates. The probability
  that they lie in the same cones of a given family is at most $1/2$, so
  the probability that the evaluation leads to $k$ steps is at most $2^{-\Omega(k)}$.
 
 
\vspace{5pt} 

\noindent\textbf{$v$ is of type (iv).} 
  Here the node $v$ needs to determine whether an input point $p_i$ lies 
  to the left of the directed line $\ell$ through input points $p_j$ and $p_k$.
 
 We partition the plane by a constant-sized family of triangles and cones, such that for each
 region $V$ in the family, the probability that the line through $p_j$ and
 $p_k$ meets $V$ is at most $1/2$.  Such a family could be constructed
 by taking a sufficiently large random sample of pairs $p_j$ and $p_k$
 and triangulating the arrangement of the lines through each pair. Such
 a construction has a non-zero probability of success, and therefore
 such a family exists. (Other than the source of the random lines used in the
 construction, this scheme goes back at least to \cite{Cla87}; a tighter
 version, called a \emph{cutting}, could also be used \cite{ChazCuttings}.) 
 
 When computing $C$, suppose $p_i$ is in region $V$ of the family.
 If the line $\overline{p_jp_k}$ does not meet $V$, then the comparison
 outcome is known immediately. This occurs with probability at least $1/2$.
 Moreover, determining the region containing $p_i$ can be done
 with a constant number of comparisons of type (i), and determining
 if $\overline{p_jp_k}$ meets $V$ can be done with a constant number of
 comparisons of type (iii); for the latter, suppose $V$ is a triangle.
 If $p_j\in V$, then $\overline{p_jp_k}$ meets $V$. Otherwise,
 suppose $p_k$ is above all the lines through $p_j$ and each
 vertex of $V$; then $\overline{p_jp_k}$ does not meet $V$. Also, if $p_k$
 is below all the lines through $p_j$ and each vertex, then $\overline{p_j p_k}$
 does not meet $V$. Otherwise, $\overline{p_j p_k}$ meets $V$.
 So a constant number of type (i) and type (iii) queries suffice.
 
 By recursively building a tree for each
 region $V$ of the family, comparisons of type (iv) can be done
 via a tree whose nodes use comparisons of type (i) and (iii) only.
 Since the probability of resolving the comparison is at least
 $1/2$ with each family of regions that is visited, the expected
 number of nodes visited is constant.
\end{proof}

\begin{proof}[of \Lem{restrictedTree}]
We transform $\cT$ into a tree $\cT'$ that
has no comparisons of type (iv), by using the construction
of  \Clm{node} where nodes of type (iv) are replaced by a tree.
We then transform $\cT'$
into a tree $\cT''$ that has no comparisons of type (iii)
or (iv), and finally transform $\cT'''$ into a restricted
tree. Each such transformation is done in the same general way,
using one case of \Clm{node}, so  we focus on the first one.

We incrementally transform $\cT$ into the tree $\cT'$.
In each such step, we have a partial restricted comparison tree 
$\cT''$ that will eventually become $\cT'$.
Furthermore, during the process each node of $\cT$ is in one of 
three different states. It is either \emph{finished},
\emph{fringe}, or \emph{untouched}. Finally, we have a function $S$
that assigns to  each finished and to each fringe node of $\cT$ a subset 
$S(v)$ of nodes in $\cT''$.

The initial situation is as follows: all nodes of $\cT$ are untouched 
except for the root which is fringe. Furthermore, the partial tree 
$\cT''$ consists of a single root node $r$ and the function
$S$ assigns the root of $\cT$ to the set  $\{r\}$.

Now our transformation proceeds as follows.
We pick a fringe node $v$ in $\cT$, and mark $v$ as finished.
For each child $v'$ of $v$, if $v'$ is an internal
node of $\cT$, we mark it as fringe. Otherwise, we mark $v'$ as finished.
Next, we apply \Clm{node} to each node $w \in S(v)$. Note that this
is a valid application of the claim, since $w$ is a node of $\cT''$, a 
restricted tree. Hence $\mathcal{R}_w$ is a product set, and the distribution
$\cD$ restricted to $\mathcal{R}_w$ is a product distribution.
Hence, replace each node $w \in S(v)$ in $\cT''$ by the subtree given 
by \Clm{node}.
Now $S(v)$ contains the roots of these subtrees. Each leaf of each such
subtree corresponds to an outcome of the comparison in $v$.
(Potentially, the subtrees are countably infinite, but the expected number
of steps to reach a leaf is constant.)
For each child
$v'$ of $v$, we define $S(v')$ as the set of all such leaves 
that correspond to the same outcome of the comparison as $v'$. 
We continue this process until there are no fringe nodes left. 
By construction, the resulting tree $\cT'$ is restricted. 

It remains
to argue that $d_{\cT'} = O(d_{\cT})$. 
Let $v$ be a node of $\cT$. 
We define two random variables $X_v$ and $Y_v$. The variable $X_v$ is 
the indicator random variable for the event that the node
$v$ is traversed for a random input $P \sim \cD$.  
The variable $Y_v$ denotes the number of nodes traversed in $\cT'$
that correspond to $v$ (i.e., the number of nodes
needed to simulate the comparison at $v$, if it occurs).
We have $d_{\cT} = \sum_{v \in \cT} \EX[X_v]$, because if the leaf corresponding
to an input $P \sim \cD$ has depth $d$, exactly $d$ nodes are traversed to
reach it.
We also have $d_{\cT'} = \sum_{v \in \cT} \EX[Y_v]$, since
each node in $\cT'$ corresponds to exactly one node $v$ in $\cT$. 
\Clm{YvBound} below 
shows that $\EX[Y_v] = O(\EX[X_v])$, which completes the 
proof.
\end{proof}

\begin{claim}\label{clm:YvBound} $\EX[Y_v] \leq c\EX[X_v]$
\end{claim}

\begin{proof}
Note that $\EX[X_v] = \Pr[X_v = 1] = \Pr[P \in \mathcal{R}_v]$.
Since the sets $\mathcal{R}_w$, $w \in S(v)$, partition $\mathcal{R}_v$, 
we can write $\EX[Y_v]$ as
\begin{multline*}
\EX[Y_v \mid X_v = 0]\Pr[X_v = 0] +\\
\sum_{w \in S(v)} \EX[Y_v \mid P \in \mathcal{R}_w]\Pr[P \in \mathcal{R}_w].
\end{multline*}
Since $Y_v = 0$ if $P \notin \mathcal{R}_v$, we have
$\EX[Y_v \mid X_v = 0] = 0$ and also
$\Pr[P \in \mathcal{R}_v] = \sum_{w \in S(v)} \Pr[P \in \mathcal{R}_w]$. 
Furthermore, by \Clm{node}, 
we have $\EX[Y_v \mid P \in \mathcal{R}_w] \leq c$.
The claim follows.
\end{proof}

\subsection{Entropy-sensitive comparison trees} \label{sec:entropy}

Since every linear comparison tree can be made restricted, we can incorporate 
the entropy
of $\cD$ into the lower bound. For this we define
entropy-sensitive trees, which are useful because the depth of
a node $v$ is related to the probability of the corresponding region
$\mathcal{R}_v$.

\vfill\eject

\begin{definition}
We call a restricted linear comparison tree \emph{entropy-sensitive}
if each comparison ``$p_i \in \ell^+?$'' is such that $\ell$
is a halving line for the current region $R_i$.
\end{definition}

\begin{lemma}\label{lem:entropyDepth}
Let $v$ be a node in an entropy-sensitive comparison tree, and let
$\mathcal{R}_v = R_1 \times R_2 \times \cdots \times R_n$.
Then $d_v = - \sum_{i=1}^n \log \Pr[R_i]$.
\end{lemma}

\begin{proof}
We use induction on the depth of $v$. For the root $r$
we have $d_r = 0$. Now, let $v'$ be the parent of $v$. Since
$\cT$ is entropy-sensitive, we reach $v$ after performing a comparison
with a halving line in $v'$. This halves the measure of exactly one
region in $\mathcal{R}_v$, so the sum increases by one.
\end{proof}

As in \Lem{restrictedTree}, we can make every restricted linear
comparison tree entropy-sensitive without affecting its expected depth too
much.
\begin{lemma}\label{lem:entropyTree}
Let $\cT$ a restricted linear comparison tree. Then there exists an 
entropy-sensitive comparison
tree $\cT'$ with expected depth $d_{\cT'} = O(d_\cT)$.
\end{lemma}

\begin{proof}
The proof extends the proof of \Lem{restrictedTree}, via
an extension to \Clm{node}.
We can regard a comparison against a fixed halving
line as simpler than an comparison against an arbitrary fixed
line. Our extension of \Clm{node} is the claim
that any type (i) node can be replaced by a tree
with constant expected depth, as follows.
A comparison $p_i \in \ell^+$ 
can be replaced by a sequence of comparisons
to halving lines.
Similar to the reduction for type (ii) comparisons in \Clm{node},
this is done by binary search. That is, let $\ell_1$ be a halving line
for $R_i$ parallel to $\ell$. We compare $p_i$ with $\ell$. If this resolves the
original comparison, we declare success. Otherwise, we repeat the process
with  the halving line for the new region $R_i'$. In each step, the probability
of success is at least $1/2$. The resulting comparison tree
has constant expected depth; we now apply the construction
of \Lem{restrictedTree} to argue that for a restricted tree $\cT$
there is an entropy-sensitive version $\cT'$ whose
expected depth is larger by at most a constant factor.
%
\end{proof}

Recall that $\text{OPT}_\cD$ is the
expected depth of an optimal linear comparison tree that
computes the maxima for $P \sim \cD$.
We now describe how to characterize $\text{OPT}_\cD$ in terms of
entropy-sensitive comparison trees. 
We first state a simple property that follows
directly from the definition of certificates and the properties of
restricted comparison trees.

\begin{prop} \label{prop:dom} Consider a leaf $v$
of a restricted linear comparison tree $\cT$ computing the maxima. 
Let $R_i$ be the region 
associated with non-maximal point $p_i \in P$ in $\mathcal{R}_v$. 
There exists some region $R_j$
associated with an extremal point $p_j$ such that every point in $R_j$
dominates every point in $R_i$.
\end{prop}

We now enhance the
notion of a certificate (\Def{cert}) to make it more
useful for our algorithm's analysis.
For technical reasons, we want points to be ``well-separated"
according to the slab structure $\bS$. By \Prop{dom},
every non-maximal point is associated with a dominating
region.

\begin{definition} \label{def:s-label}
Let $\bS$ be a slab structure. A certificate for an 
input $P$ is called \emph{$\bS$-labeled} if the following holds.
Every
maximal point is labeled with the leaf slab of $\bS$ containing it.
Every non-maximal point is either placed in the containing leaf slab,
or is separated from a dominating region by a slab boundary.
\end{definition}

We naturally extend this to trees that compute the $\bS$-labeled
maxima. 

\begin{definition}
A linear comparison tree $\cT$ \emph{computes} the $\bS$-labeled maxima
of $P$ if each leaf $v$ of  $\cT$ is labeled with a 
$\bS$-labeled certificate that is valid for every possible input
$P \in \mathcal{R}_v$.
\end{definition}

\begin{lemma} \label{lem:comp} 
There exists an entropy-sensitive comparison tree $\cT$ computing the 
$\bS$-labeled maxima whose expected depth over $\cD$ is $O(n + \textup{\OPT}_\cD)$.
\end{lemma}

\begin{proof} Start with an optimal linear comparison tree $\cT'$
that computes the maxima. At every leaf, we have a list $M$ with the maximal
points in sorted order. 
We merge $M$ with the list of slab boundaries of $\bS$ to label
each maximal point with the leaf slab of $\bS$ containing it. 
We now deal with the non-maximal points. Let $R_i$ be the region
associated with a non-maximal point $p_i$, and $R_j$ be the
dominating region. Let $\lambda$ be the leaf slab containing $R_j$.
Note that the $x$-projection of $R_i$ cannot extended to the right of $\lambda$.
If there is no slab boundary separating $R_i$ from $R_j$, then 
$R_i$ must intersect $\lambda$. With one more comparison, we can
place $p_i$ inside $\lambda$ or strictly to the left of it.
All in all, with $O(n)$ more comparisons than $\cT'$, we have a tree $\cT''$
that
computes the $\bS$-labeled maxima. Hence, the expected depth is $\OPT_\cD + O(n)$. 
Now we apply 
Lemmas~\ref{lem:restrictedTree} and \ref{lem:entropyTree} to $\cT''$ 
to get an entropy-sensitive comparison tree $\cT$ computing the 
$\bS$-labeled maxima with expected depth $O(n + \OPT_\cD)$.
\end{proof}

\section{The algorithm}
\label{sec:algorithm}

In the learning phase, the algorithm constructs a slab structure $\bS$ and
search trees $T_i$, as given in Lemmas~\ref{lem:slabstruct} and~\ref{lem:tree}.
Henceforth, we assume that we have these data structures, and will describe
the algorithm in the limiting (or stationary) phase.
Our algorithm proceeds by searching
progressively each point $p_i$ in its tree $T_i$. 
However, we need to choose the order of 
the searches carefully. 

At any stage of the algorithm, each point $p_i$ is placed in some slab 
$S_i$.  The algorithm maintains a set $A$ of \emph{active points}. An
inactive point is either proven to be non-maximal, or it has been placed in 
a leaf slab. The active points are stored in a data structure $L(A)$.
This structure is similar to a heap and supports the operations 
\textit{delete}, \textit{decrease-key}, and \textit{find-max}. The key 
associated with an active point $p_i$ is the  
right boundary of the slab $S_i$ 
(represented as an element of $[|\bS|]$).

We list the variables that the algorithm maintains.
The algorithm is initialized with $A = P$, and each $S_i$ is the largest 
slab in $\bS$.  Hence, all points have  key $|\bS|$, and we insert all these keys 
into $L(A)$. 

\medskip
\begin{asparaitem}
	\item $A, L(A)$: the list $A$ of active points stored in data structure $L(A)$.
	\item $\wlambda, B$: Let $m$ be the largest key among the active points. Then $\wlambda$ is the leaf slab 
	whose right boundary is $m$ and
	$B$ is a set of points located in $\wlambda$. Initially $B$ is empty and $m$ is $|S|$, corresponding to the
	 $+\infty$ boundary of the rightmost, infinite, slab.
	\item $M, \hat{p}$: $M$ is a sorted (partial) list of currently discovered maximal points
	and $\hat{p}$ is the leftmost among those. Initially $M$ is empty and $\hat{p}$ is a ``null'' point
	that dominates no input point.
\end{asparaitem}
\medskip

The algorithm involves a main procedure \textbf{Search}, and an auxiliary procedure
\textbf{Update}. The procedure \textbf{Search} chooses a point and proceeds
its search by a single step in the appropriate tree. Occasionally, it will invoke \textbf{Update}
to change the global variables. The algorithm repeatedly calls \textbf{Search}
until $L(A)$ is empty. 
After that, we perform a final call to \textbf{Update} in order to
process any points that might still remain in $B$.



\medskip

\noindent
\textbf{Search}.  
Let $p_i$ be obtained by performing 
a \emph{find-max} in $L(A)$. 
If the maximum key $m$ in $L(A)$ is less than the right 
boundary of $\widehat{\lambda}$, we invoke \textbf{Update}. 
If $p_i$ is dominated by $\hat{p}$, we delete $p_i$ from $L(A)$.
If not, we advance 
the search of $p_i$ in $T_i$ by a single step, if possible. 
This updates the slab $S_i$. If 
the right boundary of $S_i$ has decreased, we perform the appropriate 
\emph{decrease-key} operation on $L(A)$. (Otherwise, we do nothing.)

Suppose the point $p_i$ reaches a leaf slab $\lambda$. 
If $\lambda = \widehat{\lambda}$, we remove $p_i$ from 
$L(A)$ and insert it
in $B$ (in time $O(|B|)$). Otherwise, we leave $p_i$ in $L(A)$.

\medskip
	
\noindent
\textbf{Update}. 
We sort all the points in $B$ and update the list of current maxima.
As \Clm{order} will show, we have the sorted list of maxima
to the right of $\wlambda$. Hence, we can append to this list in $O(|B|)$
time.
%
We reset $B = \emptyset$,
set $\widehat{\lambda}$ to the leaf slab to the left of $m$, 
and return.
%
%

\medskip

We prove some preliminary claims. We state an important invariant maintained by the algorithm,
and then give a construction for the data structure $L(A)$. 

\begin{claim} \label{clm:order} 
At any time in the algorithm, 
the maxima of all points to the right of $\wlambda$ have been determined in sorted order.
\end{claim}

\begin{proof} The proof is by backward induction on $m$, the right
boundary of $\wlambda$.
When $m = |S|$, then this is trivially true. Let us assume
it is true for a given value of $m$, and trace the algorithm's
behavior until the maximum key becomes smaller than $m$ (which is done
in {\bf Update}). When {\bf Search}
processes a point $p$ with a key of $m$ then either (i) the key value decreases; 
(ii) $p$ is dominated by $\hat{p}$; or (iii) $p$ is eventually placed in 
$\wlambda$
(whose right boundary is $m$). In all cases, when the maximum key decreases
below $m$, all points in $\wlambda$ are either proven to be non-maximal
or are in $B$. By the induction hypothesis, we already have a sorted
list of maxima to the right of $m$.
The procedure {\bf Update} will sort the points in $B$
and all maximal points to the right of $m-1$ will be determined.
%
%
\end{proof}

\begin{claim} \label{clm:ds} 
Suppose there are $x$ \emph{find-max} operations and $y$ \emph{decrease-key}
operations.
We can implement the data structure $L(A)$ such that the total time 
for the operations is $O(n + x + y)$. The storage requirement is
$O(n)$.
\end{claim}

\begin{proof} 
We represent $L(A)$ as an array of lists. For every $k \in [|\bS|]$,
we keep a list of points whose key values are $k$. We maintain $m$, 
the current maximum key. The total storage is $O(n)$. A \emph{find-max} can 
trivially be done in $O(1)$ time, and an \emph{insert} is done by adding the 
element to the appropriate list.  A \emph{delete} is done by deleting the 
element from the list (supposing appropriate pointers are available). 
We now have to update the maximum. If the list 
at $m$ is non-empty, no action is required. If it is empty,
we check sequentially whether the list at $m-1, m-2, \ldots$ is empty. 
This will eventually lead to the maximum. To do a \emph{decrease-key}, we 
\emph{delete}, \emph{insert}, and then update the maximum.

Note that since all key updates are \emph{decrease-key}s, the maximum 
can only decrease.  Hence, the total overhead for scanning for a new maximum 
is $O(n)$.
\end{proof}

\subsection{Running time analysis} \label{sec:runtime}

The aim of this section is to prove the following lemma.

\begin{lemma} \label{lem:algo} 
The algorithm runs in $O(n+\textup{\OPT}_\cD)$ time.
\end{lemma}

We can easily bound the running time of all calls to {\bf Update}.

\begin{claim} \label{clm:update} The expected time for
all calls to {\bf Update} is $O(n)$.
\end{claim}

\begin{proof}
The total time taken for all calls to \textbf{Update} is at most
the time taken to sort points within leaf slabs. By \Lem{slabstruct}, 
this takes expected time
\[
  \EX \Bigl[\sum_{\lambda \in \bS} X_\lambda^2\Bigr]
=  \sum_{\lambda \in \bS} \EX\bigl[X_\lambda^2\bigr]
=  \sum_{\lambda \in \bS} O(1)
=  O(n).
\] 
\end{proof}

The important claim is the following, since it allows us to relate the time
spent by {\bf Search} to the entropy-sensitive comparison trees.
\Lem{algo} follows directly from this.

\begin{claim} \label{clm:search} Let $\mathcal{T}$ be an
entropy-sensitive comparison tree computing $\bS$-labeled maxima.
Consider a leaf $v$ labeled with the regions 
$\mathcal{R}_v = (R_1, R_2, \ldots, R_n)$,
and let $d_v$ denote the depth of $v$.
Conditioned on $P \in \mathcal{R}_v$, the expected
running time of \textup{\textbf{Search}} is 
$O(n + d_v)$.
\end{claim}

\begin{proof} 

For each $R_i$, let $S_i$ be the smallest slab of $\bS$ that completely contains 
$R_i$. We will show that the algorithm
performs at most an $S_i$-restricted search for input $P \in \mathcal{R}_v$. 
If $p_i$ is maximal, then $R_i$ is contained in a leaf slab (this
is because the output is $\bS$-labeled). Hence $S_i$ is a leaf slab
and an $S_i$-restricted search for a maximal $p_i$ is just 
a complete search.

Now consider a non-maximal $p_i$. By the properties of $\bS$-labeled
maxima, the associated region $R_i$ is either inside a leaf slab or
is separated by a slab boundary from the dominating region $R_j$.
In the former case, an $S_i$-restricted search is a complete search.
In the latter case, we argue that an $S_i$-restricted search suffices
to process $p_i$.  This follows
from \Clm{order}: by the time an $S_i$-restricted search
finishes, all maxima to the right of $S_i$ have been determined.
In particular, we have found $p_j$, and thus $\hat p$ dominates $p_i$.
Hence, the search for $p_i$ will proceed no further.


The expected search time taken conditioned on $P \in \mathcal{R}_v$ is the sum 
(over $i$) of the conditional expected $S_i$-restricted search times. 
Let $\cE_i$ denote the event that $p_i \in R_i$, and $\cE$ be the event 
that $P \in \mathcal{R}_v$.  We have $\cE = \bigwedge_i \cE_i$.
By the independence of the distributions and linearity of expectation
\begin{align*} 
&\EX_\cE[\text{search time}]\\ 
&= 
\sum_{i=1}^n \EX_\cE[\text{$S_i$-restricted search time for $p_i$}] \\
& =  \sum_{i=1}^n \EX_{\cE_i} [\text{$S_i$-restricted search time for $p_i$}].
\end{align*}
By \Lem{search-time}, the time for an $S_i$-restricted search 
conditioned on $p_i \in R_i$ is $O(-\log \Pr[p_i \in R_i] + 1)$. 
By \Lem{entropyDepth}, $d_v = \sum_i -\log \Pr[p_i \in R_i]$,
completing the proof.
\end{proof}

We can now prove the main lemma.

\begin{proof}[of \Lem{algo}] 
By \Lem{comp}, there exists
an entropy-sensitive comparison tree $\cT$ that computes the $\bS$-labeled maxima
with expected depth $O(\OPT+n)$.
According to \Clm{search}, the expected running time of {\bf Search}
is $O(\OPT+n)$. \Clm{update} tells us the expected time for {\bf Update}
is $O(n)$, and we add these bounds to complete the proof.
\end{proof}

\section{Data structures obtained during the learning phase}
\label{sec:learning}

Learning the vertical slab structure $\bS$ is very similar to
to learning the $V$-list in Ailon \etal~\cite[Lemma~3.2]{AilonCCLMS11}. 
We repeat the construction
and proof for convenience: take the union of the first 
$k = \log n$ inputs $P_1$, $P_2$, $\ldots$, $P_k$,
and sort those points by $x$-coordinates. This gives a list 
$x_0,x_1,\ldots,x_{nk-1}$. 
Take the $n$ values $x_0,x_k,x_{2k},\ldots,x_{(n-1)k}$. They define
the boundaries for $\bS$. We recall a useful 
and well-known
fact~\cite[Claim~3.3]{AilonCCLMS11}.
\begin{claim}\label{clm:indicator-square}
Let $Z = \sum_i Z_i$ be a sum of nonnegative random variables such that
$Z_i = O(1)$ for all $i$, $\EX[Z] = O(1)$, and for 
all $i,j$, $\EX[Z_iZ_j] = \EX[Z_i]\EX[Z_j]$. Then $\EX[Z^2] = O(1)$.
\end{claim}

%
Now let $\lambda$ be a leaf slab in $\bS$. Recall that we denote
by $X_\lambda$ the number of points of a random input $P$ that end
up in $\lambda$. Using \Clm{indicator-square}, we quickly
obtain the following lemma.

\begin{lemma}\label{lem:leaf-tail}
With probability $1-n^{-3}$ over the construction of $\bS$, we have
$\EX[X_\lambda^2] = O(1)$ for all leaf slabs $\lambda \in \bS$.
\end{lemma}

\begin{proof}
Consider two values $x_i$, $x_j$ from the original
list. Note that all the other $kn - 2$ values are independent of these 
two points.
For every $r \notin \{i,j\}$, let $Y^{(r)}_t$ be the indicator random
variable for $x_r \in t \eqdef [x_i, x_j)$. Let 
$Y_t = \sum_r Y^{(r)}_t$.
Since the $Y^{(r)}_t$'s are independent, by Chernoff's bound~\cite{AlonSp00},
for any $\beta \in (0,1]$,  
\[
\Pr[Y_t \leq (1-\beta)\EX[Y_t]] \leq \exp(-\beta^2\EX[Y_t]/2).
\]
 With probability at least $1 - n^{-5}$, if $\EX[Y_t] > 12\log n$,
then $Y_t > \log n$. By applying the same argument
for any pair $x_i, x_j$ and taking a union bound over all
pairs, with probability at least $1 - n^{-3}$ the following holds:
for any  pair $t$, if $Y_t \leq \log n$, then $\EX[Y_t] \leq 12\log n$. 

For any leaf slab $\lambda = [x_{ak},x_{(a+1)k}]$, we have
$Y_\lambda \leq \log n$.
Let $X^{(i)}_\lambda$ be the indicator
random variable for the event that $x_i \sim \cD_i$
lies in $\lambda$, so that $X_\lambda = \sum_i X^{(i)}_\lambda$.
Since $\EX[Y_\lambda] \geq (\log n - 2) \EX[X_\lambda]$, we get
$\EX[X_\lambda] = O(1)$. By independence of the ${\cD}_i$'s, 
for all $i,j$, $\EX\bigl[X^{(i)}_\lambda X^{(j)}_\lambda\bigr] = 
\EX\bigl[X^{(i)}_\lambda\bigr] \EX\bigl[X^{(j)}_\lambda\bigr]$,
so $\EX[X^2_\lambda] = O(1)$, by \Clm{indicator-square}.
\end{proof}

\Lem{slabstruct} follows immediately from 
\Lem{leaf-tail} and the fact that sorting the
$k$ inputs $P_1$, $P_2$, $\ldots$, $P_k$ takes $O(n \log^2 n)$
time.
After the leaf slabs have been determined,
the search trees $T_i$ can be found using
essentially the same techniques as before~\cite[Section~3.2]{AilonCCLMS11}.
The main idea is to use $n^\eps\log n$ rounds to find the first
$\eps \log n$ levels of $T_i$, and to use
a balanced search tree for searches that need to
proceed to a deeper level. This only costs a factor of $\eps^{-1}$.
We restate \Lem{tree} for convenience.

\begin{lemma}
Let $\eps > 0$ be a fixed parameter.
In $O(n^{\eps})$ rounds and $O(n^{1+\eps})$ time, we can construct search trees
$T_1$, $T_2$, $\ldots$, $T_n$ over $\bS$ such that the following holds.
\textup(i\textup) the trees can be totally represented in $O(n^{1+\eps})$ space;
\textup(ii\textup) probability $1-n^{-3}$ over the construction of the
$T_i$s: every $T_i$ is $O(1/\eps)$-optimal for restricted searches over $\cD_i$.
\end{lemma}

\begin{proof}
Let $\delta > 0$ be some sufficiently small constant and $c$
be sufficiently large . For $k = c\delta^{-2}n^{\eps}\log n$ rounds and each $p_i$,
we record the leaf slab of $\bS$ that contains it. We break the
proof into smaller claims.

\begin{claim} \label{clm:prob} Using $k$ inputs, we can
compute estimates $\hat{q}(i,S)$ for each index $i$ and slab $S$.
The following guarantee holds (for all $i$ and $S$)
with probability $>1 - 1/n^3$ over the choice of the $k$ inputs. 
If at least $5\log n$ instances of $p_i$ fell in $S$,
then $\hat{q}(i,S) \in  
[(1-\delta)q(i,S),(1+\delta)q(i,S)]$\footnote{We remind the reader that this the
probability that $p_i \in S$.}.
\end{claim}

\begin{proof}
For a slab $S$, let $N(S)$ be the number of times $p_i$
was in $S$, and let $\hat{q}(i,S) = N(S)/k$ be the
empirical probability for this event
($\hat{q}(i,S)$ is an estimate  of 
$q(i,S)$). 
Fix a slab $S$. If $q(i,S) \leq 1/2n^{\eps}$, then
by a Chernoff bound we get
$\Pr[N(S) \geq 5\log n \geq 10kq(i,S)] \leq 2^{-5\log n} = n^{-5}$.
Furthermore, if $q(i,S) \geq 1/2n^{\eps}$, then $q(i,S)k \geq (c/2\delta^2)\log n$ and
$\Pr[N(S) \leq (1-\delta)q(i,S)k] \leq \exp(-q(i,S)\delta^2k/4) \leq  n^{-5}$ as
well as
$\Pr[N(S) \geq (1+\delta)q(i,S)k] \leq \exp(-\delta^2q(i,S)k/4) \leq n^{-5}$.
Thus, by taking a union bound, we get that with probability at
least $1-n^{-3}$ for any slab $S$, if $N(S) \geq 5\log n$,
then $q(i,S) \geq n^{-\eps}/2$ and 
hence $\hat{q}(i,S) \in  [(1-\delta)q(i,S),(1+\delta)q(i,S)]$.
\end{proof}

We will henceforth assume that this claims holds for all $i$ and $S$.
Based on the values $\hat{q}(i,S)$, we construct the search trees.
The tree $T_i$ is constructed recursively.
We will first create a partial search tree, where
some searches may end in non-leaf slabs (or, in other words,
leaves of the tree may not be leaf slabs).
The root is the just the largest slab. Given a slab
$S$, we describe how the create the sub-tree of $T_i$
rooted at $S$. If $N(S) < 5\log n$, then we make
$S$ a leaf. Otherwise, we pick a leaf slab $\lambda$ such that for the slab $S_l$
consisting of all leaf slabs (strictly) to the left of $\lambda$ and
the slab $S_r$ consisting of all leaf slabs (strictly) to the right
of $\lambda$ we have $\hat{q}(i,S_l) \leq (2/3)\hat{q}(i,S)$ and
$\hat{q}(i,S_r) \leq (2/3)\hat{q}(i,S)$. We make $\lambda$ a leaf child of $S$. Then we
recursively create trees for $S_l$ and $S_r$ and attach them as children to $S$. 
For any internal node of the tree $S$, we have $q(i,S) \geq n^\eps/2$,
and hence the depth is at most $O(\eps \log n)$. Furthermore,
this partial tree is $\beta$-reducing (for some constant $\beta$).
The partial tree $T_i$ is extended to a complete tree
in a simple way. From each $T_i$-leaf that is not a leaf slab, we perform
a basic binary search for the leaf slab. 
This yields a tree $T_i$ of depth at most $(1+O(\eps))\log n$.
Note that we only need to store the partial $T_i$ tree, and
hence the total space is $O(n^{1+\eps})$.

Let us construct, as a thought experiment,
a related tree $T'_i$. Start with the partial $T_i$.
For every leaf that is not a leaf slab, extend it downward
using the true probabilities $q(i,S)$. In other words,
let us construct the subtree rooted at a new node $S$
in the following manner. We pick a leaf slab $\lambda$ such that
$q(i,S_l) \leq (2/3)q(i,S)$ and
$q(i,S_r) \leq (2/3)q(i,S)$ (where $S_l$ and $S_r$ are as defined
above). This ensures that $T'_i$ is $\beta$-reducing.
By \Lem{search-time}, $T'_i$ is $O(1)$-optimal
for restricted searches over $\cD_i$ (we absorb the $\beta$
into $O(1)$ for convenience).

\begin{claim} \label{clm:1/eps} The tree $T_i$ is $O(1/\eps)$-optimal
for restricted searches. 
\end{claim}

\begin{proof}
Fix a slab $S$ and an $S$-restricted
distribution $\cD_S$. Let $q'(i,\lambda)$ (for each leaf slab $\lambda$) 
be the series of values defining $\cD_S$. Note that $q'(i,S) \leq q(i,S)$.
Suppose $q'(i,S) \leq n^{-\eps/2}$. Then $-\log q'(i,S) \geq \eps (\log n)/2$.
Since any search in $T_i$ takes at most $(1+O(\eps))\log n$ steps,
the search time is at most $O(\eps^{-1}(-\log q'(i,S) + 1))$. 

Suppose $q'(i,S) > n^{-\eps/2}$. Consider a single search for some $p_i$.
We will classify this search based on the leaf of the partial tree
that is encountered. By the construction of $T_i$, any leaf $S'$ is either a leaf slab or
has the property that $q(i,S') \leq n^{-\eps}/2$.
The search is of \emph{Type 1} if the leaf
of the partial tree actually represents a leaf slab (and hence
the search terminates). The search is of \emph{Type 2} (resp. \emph{Type 3}) if the 
leaf of the partial tree is a slab $S$ is an internal node of $T_i$
and the depth is at least (resp. less than) $\eps(\log n)/3$. 

When the search is of Type 1, it is identical in both $T_i$ and $T'_i$. When
the search is of Type 2, it takes at $\eps(\log n)/3$ in $T'_i$ and at most
(trivially) $(1+O(\eps))(\log n)$ in $T_i$. The total number of leaves (that are not leaf slabs)
of the partial tree at depth less than $\eps(\log n)/3$ is at most $n^{\eps/3}$.
The total probability mass of $\cD_i$ inside such leaves is at most $n^{\eps/3}\times n^{-\eps}/2 < n^{-2\eps/3}$.
Since $q'(i,S) > n^{-\eps/2}$, in the restricted distribution $\cD_S$, the probability
of a Type 3 search is at most $n^{-\eps/6}$.

Choose a random $p \sim \cD_S$.
Let $\cE$ denote the event that a Type 3 search occurs.
Furthermore, let $X_p$ denote the depth of the search in $T_i$ and $X'_p$ denote
the depth in $T'_i$. When $\cE$ does not occur, we have argued that $X_p \leq O(X'_p/\eps)$.
Also, $\Pr(\cE) \leq n^{-\eps/6}$. The expected search time is just $\EX[X_p]$.
By Bayes' rule,
\begin{align*}
	\EX[X_p] & =  \Pr(\overline{\cE}) \EX_{\overline{\cE}}[X_p] + \Pr(\cE)\EX_{\cE}[X_p] \\
	& \leq  O(\eps^{-1}\EX_{\overline{\cE}}[X'_p]) +  n^{-\eps/6}(1+O(\eps))\log n \\
	\EX[X'_p] & =  \Pr(\overline{\cE}) \EX_{\overline{\cE}}[X'_p] + \Pr(\cE)\EX_{\cE}[X_p] \\
	  \Longrightarrow\quad &\EX_{\overline{\cE}}[X'_p] \leq \EX[X'_p]/\Pr(\overline{\cE}) \leq 2\EX[X'_p] 
\end{align*}
Combining, the expected search time is $O(\eps^{-1}(\EX[X'_p] + 1))$. Since $T'_i$
is $O(1)$-optimal for restricted searches, $T_i$ is $O(\eps^{-1})$-optimal.
\end{proof}

\end{proof}



%

%

\section{Acknowledgments}
C. Seshadhri was funded by the Early-Career LDRD program at Sandia National Laboratories.
Sandia National Laboratories is a multi-program laboratory managed and operated by Sandia Corporation, a wholly owned subsidiary of Lockheed Martin Corporation, for the U.S. Department of Energy's National Nuclear Security Administration under contract DE-AC04-94AL85000.

We would like to thank Eden Chlamt\'{a}\v{c} for suggesting a simple proof for 
\Clm{YvBound}.

\bibliographystyle{alpha}
\bibliography{si-extrema-full}

\full{
\appendix

\section{Restricted searches}\label{app:restricted}

\begin{lemma} 
Given an interval $S$ in $\bU$, 
let $\cF_S$ be an $S$-restricted distribution of $\cF$. Let $T$
be a $\mu$-reducing search tree for $\cF$.
Conditioned on $j$
drawn from $\cF_S$, the expected time of an $S$-restricted search in 
$T$ for $j$ is at most $(b/\log(1/\mu))(-\log q'_S + 1)$ (for some absolute constant $b$).
\end{lemma}

Now that we may assume that we are comparing against an entropy
sensitive comparison tree, we need to think about how to make our
searches entropy-sensitive. For this we proceed as follows.
By \Lem{slabstruct}, we have a vertical slab structure
$\bS$ such that each leaf slab contains only constantly many
points in expectation. Now, for each distribution $\cD_i$, 
we construct an optimal search tree $T_i$ for the leaf
slabs of $\bS$. 
The recursion
continues until $S_l$ or $S_r$ are empty. The search in
$T_i$ proceeds in the obvious way. To find the leaf
slab containing $p_i$, we begin in at the root and
check whether $p_i$ is contained in the corresponding
leaf slab. If yes, the search stop. Otherwise, we branch
to the appropriate child and continue.

Each node in $T_i$ corresponds to a slab in $\bS$, and
it is easily seen that if a node has depth $d$, then
$p_i$ is contained in the corresponding slab with 
probability at most $2^{-d}$. From this, it
quickly follows that $T_i$ is an asymptotically optimal
search tree for $\cD_i$. However, below we require a stronger
result.  Namely, we need a technical lemma showing how an optimal
search tree for some distribution $\cF$ is also useful
for some conditional distributions. 

Let $\bU$ be an ordered set and $\cF$ be a distribution
over $\bU$. For any element $j \in \bU$, we let $p_j$
denote the probability of $j$ in $\cF$. For any interval
$S$ of $\bU$, the total probability of $S$ is $p_S$.

Let $T$ be a search tree over $\bU$ with the
following properties.
For any internal node $S$ and a non-leaf child $S'$,
$p_{S'} \leq \mu p_S$. As a result, if $S$ has depth $k$,
then $p_S \leq \mu^k$.
Every internal node of $T$ has at most $2$ internal children
and at most $2$ children that are leaves.

\begin{definition} \label{def:rest-app}
Given an distribution $\cF$ and interval $S$, an \emph{$S$-restricted 
distribution $\cF_S$} is a conditional distribution of $\cF$ such that $i$ 
chosen from $\cF_S$ always falls in $S$.
\end{definition}

For any $S$-restriction $\cF_S$ of $\cF$, there exist
values $p'_j$ with the following properties.
For each $j \in S$, $p'_j \leq p_j$.
For every other $j$, $p'_j = 0$. 
The probability
of element $j$ in $\cF_S$ is $p'_j / \sum_r p'_r$. 
Henceforth, we will use the primed values to denote
the probabilities in $\cF_S$.
For interval $R$,
we set $p'_R = \sum_{r \in R} p'_r$.
Suppose we perform a search for $j \in S$.
This search is called \emph{$S$-restricted} if it terminates once we locate
$j$ in any interval contained in $S$. 

\begin{lemma} \label{lem:search-time-app} Given an interval $S$ in $\bU$, 
let $\cF_S$ be an $S$-restricted distribution.
Conditioned on $j$
drawn from $\cF_S$, the expected time of an $S$-restricted search in 
$T$ for $j$ is $O(-\log p'_S + 1)$.
\end{lemma}

\begin{proof}
We bound the number of visited nodes in an $S$-restricted search. 
We will prove, by induction 
on the distance from a leaf, that for all
visited nodes $V$ with $p_V \le 1/2$, the expected number of visited nodes 
below $V$ is
$c_1 + c\log(p_V/p'_V),$
for constants $c,c_1$. This bound clearly holds for leaves. Moreover, since 
for $V$ at depth $k$, $p_V \le \mu^k$, we have $p_V \le 1/2$ for all but 
the root and at most $1/\log(1/\mu)$ nodes below it on the search path.

We now examine all possible paths down $T$ that an $S$-restricted search
can lead to.
It will be helpful to consider the possible ways that $S$ can intersect the
nodes (intervals) that are visited in a search.
Say that the intersection $S \cap V$ of $S$ with interval $V$ is \emph{trivial}
if it is either empty, $S$, or $V$. Say that it is \emph{anchored} if it shares
at least one boundary line with $S$. 
Suppose $S \cap V = V$. Then the search will terminate at $V$, since we have
certified that $j \in S$. Suppose $S \cap V = S$, so $S$ is contained in $V$.
There can be at most one child of $V$ that contains $S$. If such a child exists,
then the search will simply continue to this child. If not, then all possible
children (to which the search can proceed to) are anchored. The search
can possibly continue to any child, at most two of which are internal nodes.
Suppose $V$ is anchored. Then at most one child of $V$ can be anchored with $S$.
Any other child that intersects $S$ must be contained in it. Refer to Figure~\ref{fig:anchored}.
%
\begin{figure}
\begin{center}
\includegraphics{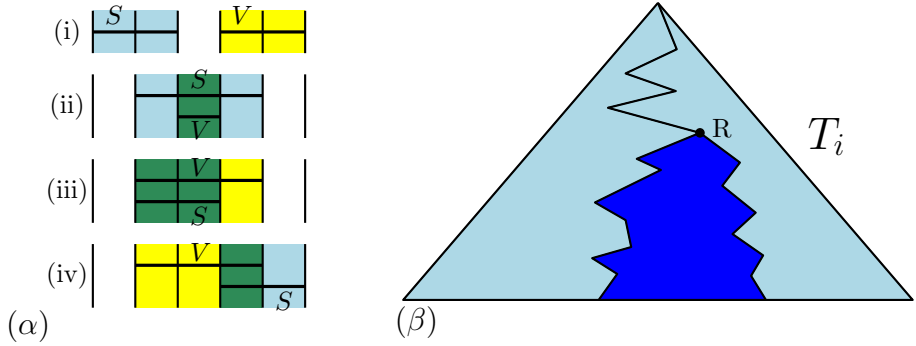}
\end{center}
\caption{($\alpha$) The intersections $S \cap V$ in (i)-(iii) are trivial, the
intersections in (iii) and (iv) are anchored; ($\beta$) every node of $T_i$
has at most one non-trivial child, except for $R$.}
\label{fig:anchored}
\end{figure}

Consider the set of all possible nodes that can be visited by an $S$-restricted search
(remove all nodes that are terminal, i.e., completely contained in $S$). These 
form a set of paths, that form some subtree of $S$. In this subtree, there is only one possible node
that has two children. This comes from some node $R$ that contains $S$ and has
two anchored (non-leaf) children. Every other node of this subtree has
a single child. Again, refer to Figure~\ref{fig:anchored}.
%

From the above, it follows that for all visited nodes $V$ with $V\ne R$, there
is at most one child $W$ whose intersection with $S$ is neither empty nor
$W$. Let $\vis(V)$ be the expected number of nodes visited below $V$, conditioned
on $V$ being visited. 
We have $\vis(V) \le 1 + \vis(W)p'_W/p'_V$, using the fact that when a search
for $j$ shows that it is contained in a node contained in $S$, the $S$-restricted
search is complete.

\begin{claim}\label{clm:WA}
For $V,W$ as above, with $p_V\le 1/2$,
if $\vis(W)\le c_1 + c\log(p_W/p'_W)$,
then for $c\ge c_1/\log(1/\mu)$,
with $\mu\in (0,1)$,
\begin{align} \label{T_i recur}
\vis(V) \le 1 + c\log(p_V/p'_V).
\end{align}
\end{claim}

\begin{proof}
By hypothesis, using $p_W\le \mu p_V$, and letting $\beta := p'_W/p'_V\le 1$,
$\vis(V)$ is no more than
\begin{multline*}
1 + (c_1 + c\log (p_W/p'_W))p'_W/p'_V 
\le  1 + (c_1 + c\log (p_V/p'_W) + c\log(\mu))\beta  \\ 
=  1 + c_1\beta + c\log(p_V)\beta 
 + c\log(1/p'_W)\beta + c\log(\mu)\beta.
\end{multline*}
The function $x\log(1/x)$ is increasing in the range $x \in (0,1/2)$. Hence,
$p'_W\log(1/p'_W)\le p'_V\log(1/p'_V)$ for  $p'_V\le p_V\le 1/2$. Since
$\beta\le 1$, we have
\begin{multline*}
\vis(V) \le  1 + c_1\beta + c\log (p_V) 
  + c\log(1/p'_V) + c\log(\mu)\beta \\
   =  1 + c\log (p_V/p'_V) + \beta(c_1 + c\log(\mu)) 
   \leq  1 + c\log (p_V/p'_V),
\end{multline*}
for $c \ge c_1/\log(1/\mu)$.
\end{proof}

Only a slightly weaker statement can be made for the node $R$ having two
nontrivial intersections at child nodes $R_1$ and $R_2$.

\begin{claim}\label{clm:WR}
For $R,R_1,R_2$ as above, if $\vis(R_i)\le c_1 + c\log (p_{R_i}/p'_{R_i})$, 
for $i=1,2$, then for $c\ge c_1/\log(1/\mu)$,
\[
\vis(R) \le 1 + c\log(p_R/p'_R) + c.
\]
\end{claim}

\begin{proof}
We have
\[
\vis(R) \le 1 + \vis(R_1)p'_{R_1}/p'_R + \vis(R_2)p'_{R_2}/p'_R.
\]
Let $\beta := (p'_{R_1} + p'_{R_2})/p'_R$. With the given bounds for $\vis(R_i)$,
then using $p_{R_i}\le \mu p_R$, $\vis(R)$ is bounded by
\begin{multline*}
1 + \sum_{i=1,2} 
   [c_1 + c\log (p_{R_i}/p'_{R_i})]p'_{R_i}/p'_R \\
\leq  1 + c_1\beta + c\beta\log(\mu) + c\beta\log(p_R) 
  + c\sum_{i=1,2} (p'_{R_i}/p'_R)\log(1/p'_{R_i}).
\end{multline*}
The sum takes its maximum value when each $p'_{R_i} = p'_R/2$, yielding
\begin{multline*}
\vis(R) 
\leq 1 + c_1\beta + c\beta\log(\mu) + c\beta\log(p_R) + c\beta\log(2/p'_R)\\
\leq 1 + c\log(p_R/p'_R) +  \beta(c_1 + c\log(\mu)) + c\log(2)
\leq 1 + c\log(p_R/p'_R) + c\log(2),
\end{multline*}
for $c\ge c_1/\log(1/\mu)$, as in (\ref{T_i recur}), except for the addition of $c\log 2 = c$.
\end{proof}

Now to complete the proof of \Lem{search-time}. For the visited nodes
below $R$, we may inductively take $c_1 = 1$ and $c=1/\log(1/\mu)$, using 
\Clm{WA}.
The hypothesis of \Clm{WR}  then holds for $R$.
For the visited node just above $R$, we may apply \Clm{WA} with
$c_1 = 1 + 1/\log(1/\mu)$ and $c\ge c_1/\log(1/\mu)$. The result is that 
for the node
$V$ just above $R$, $\vis(V)\le 1 + c\log(p_1/p'_V)$. This bound holds 
then inductively
(with the given value of $c$) for nodes further up the
tree, at least up until the $1+1/\log(1/\mu)$ top nodes. 
For the root $Q$, note that $p'_Q = p'_S$.
Thus the
expected number of visited nodes below $Q$ is at most
$1/\log(1/\mu) + 1 + c\log(p_Q/p'_Q)
    = O(1 - \log(p'_S)),$
as desired.
\end{proof}
}

\end{document}